\documentclass[runningheads,letterpaper]{llncs}

\newcommand{\longversion}[1]{}
\newcommand{\shortversion}[1]{#1}
\newcommand{\longshort}[2]{\longversion{#1}\shortversion{#2}}

\usepackage{amsmath,amssymb,amstext}
\usepackage{thmtools,thm-restate}

\usepackage{tikz}
\usepackage{boxedminipage}
\usepackage{xspace}

\spnewtheorem{observation}{Observation}{\bfseries}{\itshape}

\newcommand{\NP}{\text{\normalfont NP}}

\newcommand{\FPT}{\text{\normalfont FPT}}
\newcommand{\XP}{\text{\normalfont XP}}
\newcommand{\W}[1][xxxx]{\text{\normalfont W}[#1]}

\newcommand{\ig}{interval graph\xspace}
\newcommand{\kgig}{$k$-gap interval graph\xspace}
\newcommand{\kgigs}{$k$-gap interval graphs\xspace}
\newcommand{\mir}{multiple interval representation\xspace}
\newcommand{\mirfirst}{multiple interval representation\xspace}
\newcommand{\gap}{\mathsf{gap}}
\newcommand{\ipkvg}{interval+$kv$ graph\xspace}
\newcommand{\ipkvgs}{in\-ter\-val\-+$kv$ graphs\xspace}
\newcommand{\idsfirst}{in\-ter\-val de\-le\-tion set\xspace}
\newcommand{\ids}{in\-ter\-val de\-le\-tion set\xspace}

\newcommand{\set}[1]{\left\{ #1 \right\}}
\newcommand{\myiff}{if and only if\xspace}

\newcommand{\etal}{\emph{et al.}\xspace}
\newcommand{\mc}{\mathcal}

\newcommand{\is}{in\-de\-pen\-dent set\xspace}
\newcommand{\IS}{\textsc{In\-de\-pen\-dent Set}\xspace}
\newcommand{\clique}{clique\xspace}
\newcommand{\CLIQUE}{\textsc{Clique}\xspace}
\newcommand{\MIT}{\textsc{Mul\-ti\-ple In\-ter\-val Trans\-ver\-sal}\xspace}
\newcommand{\ds}{do\-mi\-na\-ting set\xspace}
\newcommand{\DS}{\textsc{Do\-mi\-na\-ting Set}\xspace}
\newcommand{\IDS}{\textsc{In\-de\-pen\-dent Do\-mi\-na\-ting Set}\xspace}
\newcommand{\cc}{clique co\-ver\xspace}
\newcommand{\CC}{\textsc{Clique Co\-ver}\xspace}
\newcommand{\fvs}{feed\-back ver\-tex set\xspace}
\newcommand{\FVS}{\textsc{Feed\-back Ver\-tex Set}\xspace}
\newcommand{\COL}{\textsc{Co\-lo\-ring}\xspace}
\newcommand{\HC}{\textsc{Ha\-mil\-to\-nian Cycle}\xspace}

\newcommand{\tw}{{\mathbf{tw}}}
\newcommand{\pw}{{\mathbf{pw}}}

\title{$k$-Gap Interval Graphs \thanks{%
Serge Gaspers and Stefan Szeider acknowledge support from the European Research Council
(COMPLEX REASON, 239962). Petr Golovach acknowledges the support by EPSRC (EP/G043434/1), Royal Society (JP100692). Karol Suchan acknowledges support from Conicyt Chile (Anillo ACT-88, Basal-CMM, Fondecyt 11090390).%
}}

\author{%
Fedor V. Fomin\inst{1} \and
Serge Gaspers\inst{2} \and
Petr Golovach\inst{3} \and 
Karol Suchan\inst{4,} \inst{5} \and 
Stefan Szeider\inst{2} \and 
Erik Jan van Leeuwen\inst{1} \and 
Martin Vatshelle\inst{1} \and 
Yngve Villanger\inst{1}}

\authorrunning{Fomin \etal}

\institute{%
 Department of Informatics, University of Bergen, Bergen, Norway.\\
 \longshort{\email{\{fedor.fomin, e.j.van.leeuwen, martin.vatshelle, yngve.villanger\}@ii.uib.no}}{\email{\{fomin, e.j.van.leeuwen, vatshelle, yngvev\}@ii.uib.no}}
\and
 Inst. of Information Systems, Vienna University of Technology, Vienna, Austria.\\
 \email{gaspers@kr.tuwien.ac.at}, \email{stefan@szeider.net}
\and
 School of Engineering and Computing Sciences, Durham University, Durham, UK.\\
 \email{petr.golovach@durham.ac.uk}
\and
 Facultad de Ingenier{\'i}a y Ciencias, Universidad Adolfo Ib{\'a}{\~n}ez, Santiago, Chile.\\
 \email{karol.suchan@uai.cl}
\and
 \longshort{Faculty of Applied Mathematics}{WMS}, AGH - University of Science and Technology, Krakow, Poland.
}

\begin{document}

\maketitle

\begin{abstract}
We initiate the study of a new parameterization of graph problems.
In a \mir of a graph, each vertex is associated to at least one interval of the
real line, with an edge between two vertices \myiff an interval associated to
one vertex
has a nonempty intersection with an interval associated to the other vertex. A
graph on $n$ vertices is a $k$-gap interval graph if it has a \mir with at most
$n+k$ intervals
in total. In order to scale up the nice algorithmic properties of interval
graphs (where $k=0$), we parameterize graph problems by $k$, and find \FPT\
algorithms for several problems, including \FVS, \DS, \IS, \CLIQUE, \CC, and
\MIT. The \COL problem turns out to be \W[1]-hard and we design
an \XP\ algorithm for the recognition problem.
\end{abstract}

\section{Introduction}
A \mirfirst $f$ of a graph $G=(V,E)$ is a mapping which assigns to each vertex of $G$ a
non-empty collection of intervals on the real line so that two distinct vertices
$u$ and $v$ are adjacent \myiff there are intervals $I \in f(u)$ and $J \in
f(v)$ with $I \cap J \neq \emptyset$. For a vertex $v$, $|f(v)|$ denotes the
number of intervals that $f$ assigns to $v$. The \emph{interval number}
of $G$ is defined as
\begin{align*}
i(G) = \min \set{\max_{v\in V} \set{ |f(v)|} : f \text{ is a \mir of }
G}\enspace.
\end{align*}
The \emph{total interval number} of a graph $G=(V,E)$ is defined as
\begin{align*}
I(G) = \min \set{\sum_{v\in V} \set{ |f(v)|} : f \text{ is a \mir of }
G}\enspace.
\end{align*}
The class of \emph{$t$-interval graphs} is defined as the class of all graphs
$G$ with $i(G)\le t$.
This natural generalization of interval graphs was independently introduced by
Trotter and Harary
\cite{TrotterH79}, and by Griggs and West \cite{GriggsW80}.

Even for small fixed $t\ge 2$, these graph classes are much richer than interval
graphs. For example, the class of $2$-interval graphs
includes circular-arc graphs, outerplanar graphs, cubic graphs, and line graphs,
and the class of $3$-interval graphs includes all planar graphs
\cite{ScheinermanW83}.
Unfortunately, many problems remain \NP-hard on $2$-interval graphs (for
example, their recognition \cite{WestS84}, 3-\COL, \DS,
\IS, and \HC) or $3$-interval graphs (for example \CLIQUE
\cite{ButmanHLR10}, whose
complexity on $2$-interval graphs is open \cite{ButmanHLR10,Spinrad03}).
Parameterized by solution size, \IS, \DS, and \IDS
are $W[1]$-hard on $2$-interval graphs, even when all intervals have unit
length, whereas \CLIQUE is FPT \cite{FellowsHRV09}.

\longshort{\medskip}{\smallskip}

With the objective to generalize interval graphs while maintaining their nice
algorithmic properties, we define
\emph{$k$-gap interval graphs} as graphs that have a \mir whose total number of
\longversion{individual }intervals exceeds the
number of vertices by at most $k$. Parameterizing problems by $k$ becomes then a
reasonable attempt to
scale up the nice algorithmic properties of interval graphs to more general
graphs.

\begin{definition}
 A graph $G$ on $n$ vertices is a \emph{\kgig{}} if $I(G) \le n+k$.
\end{definition}
Throughout this paper, we assume that problems that have a \kgig as input also
have access to the corresponding \mir.
An alternative definition of \kgigs is via the identification operation.

\begin{definition}
 Given a graph $G=(V,E)$ and two vertices $a,b\in V$, the graph obtained from
$G$ by \emph{identifying} $a$ and $b$
 is the graph obtained from $G-b$ by adding all edges $\set{va : vb\in E}$.
\end{definition}

\begin{observation}\label{obs:ident}
 A graph is a \kgig \myiff it can be obtained from an \ig by a sequence of at
most $k$
 operations of identifying pairs of vertices.
\end{observation}

\noindent
Several of our \FPT\ results do not make use of the special structure of the vertices with gaps,
and also hold for the graph class interval+$kv$.

\begin{definition}
 A graph $G=(V,E)$ is an \emph{\ipkvg{}} if there is a vertex set
 $X\subseteq V$, with $|X|\le k$, such that $G\setminus X$ is an interval graph.
 We refer to such a vertex set $X$ as the \emph{\idsfirst{}} of $G$.
\end{definition}

\noindent
When discussing the complexity of problems on \ipkvgs, we assume that an \ids
is provided as a part of the input, as it is an open
question whether \textsc{Interval Vertex Deletion} is \FPT~\cite{Marx10}.
As the set of vertices that are associated to more that one interval in a \mir is an
\ids, \FPT\ results for \ipkvgs translate to \FPT\ results for \kgigs.
When the generalization of the \FPT\ result for \kgigs to \ipkvgs is straightforward,
we state the stronger result.

\paragraph{Related work.}
The class of $t$-interval graphs has been studied from the view point of
approximation algorithms \cite{Bar-YehudaHNSS06,Bar-YehudaR06,ButmanHLR10},
graph theory
\longshort{\cite{Andreae86,Andreae87,BaloghOP04,BaloghP99,BaloghP02,ChenCW04a,ErdosW85,Griggs79,GriggsW80,Gyarfas85,HopkinsTW84,Maas84,Pluhar95,Scheinerman87,Scheinerman88,Scheinerman90,ScheinermanW83,SpinradVW87,TrotterH79,West89}}{%
(see, for example \cite{Andreae86,BaloghOP04} \cite{ErdosW85,GriggsW80,ScheinermanW83,TrotterH79,West89} and references)},
classical complexity \cite{WestS84}, and parameterized complexity
\cite{BlinFV07,FellowsHRV09,JiangZ11,JiangZ11b}.
Applications have been identified in scheduling and resource allocation
\cite{Bar-YehudaHNSS06,Bar-YehudaR06},
communication protocols \cite{ButmanHLR10}, computational biology
\cite{AumannLMPY05,BafnaNR96,BlinFV07,ChenYY07,CrochemoreHLRV08,GambetteV07,Vialette04,Vialette08},
and monitoring \cite{ButmanHLR10}.
The total interval number was introduced in \cite{GriggsW80} and studied in
\cite{AndreaeA89,Catlin92,ChenC02,KostochkaW97,KratzkeW93,KratzkeW96,Raychaudhuri95}.

\paragraph{Notation and definitions.}
Let $G=(V,E)$ be a graph, $u\in V$ be a vertex, and $S\subseteq V$ be a subset
of vertices. The \emph{open neighborhood} of $v$ is defined as $N(u) = \set{v :
uv\in E}$,
its \emph{closed neighborhood} is $N[u] = N(u) \cup \set{u}$, and its degree is
$d(u)=|N(u)|$. Also, $N(S) = \bigcup_{v\in S} N(v) \setminus S$ and
$N[S] = N(S) \cup S$.
The graph $G\setminus S$ is obtained from $G$ by removing all vertices
in $S$ and all edges that have at least one endpoint in $S$.
The graph induced on $S$ is $G[S] = G
\setminus (V \setminus S)$.
By $n$ and $m$ we generally denote the number of vertices and edges of $G$.

In a \kgig $G=(V,E)$ with \mir $f$ we say that a vertex $v \in V$ has a gap if
$|f(v)|\ge 2$. We denote by $\gap_f(G)$ the set of
vertices that have gaps and omit the subscript if the context specifies the
\mir. We say that a \mir of $G$ has $k$ gaps if $\sum_{v\in V} |f(v)| = |V|+k$.

A \emph{tree decomposition} of a graph $G$ is a pair $(B,T)$ where $T$
is a tree and $B=\{B_{i} \mid i\in V(T)\}$ is a collection of subsets (called
{\em bags})
of $V(G)$ such that:
\begin{enumerate}
\item $\bigcup_{i \in V(T)} B_{i} = V(G)$, 
\item for each edge $\{x,y\} \in E(G)$, $x,y\in B_i$ for some  $i\in V(T)$, and 
\item for each $x\in V(G)$ the set $\{ i \in V(T) : x \in B_{i} \}$ induces a
connected subtree of $T$.
\end{enumerate}
The \emph{width} of a tree decomposition $(\{ B_{i} \mid i \in V(T) \},T)$ is
$\max_{i \in V(T)}\,\{|B_{i}| - 1\}$. The \emph{treewidth} of a graph $G$
(denoted $\tw(G)$) is the minimum width over all tree decompositions of $G$. 
If, in the above definitions, we restrict \longversion{the tree }$T$ to be a path
then we define the notions of \emph{path decomposition} and {\em
pathwidth} of $G$ (denoted $\pw(G)$).

We refer to \cite{DowneyF99,FlumG06,Niedermeier06} for texts on parameterized
complexity, the theoretical framework of our investigations.
\shortversion{Proofs of statements marked with \textbf{($\star$)} can be found in Appendix \ref{sec:appOmitted}.}

\section{Recognition}

The problem of recognizing \kgigs is to determine whether for a graph $G$ on $n$ vertices, $I(G) \leq n + k$.
If $k$ is part of the input, the problem is \NP-hard, as it is \NP-hard to decide whether $I(G) \leq n + (m+1-n)$, even if $G$ is a planar, $3$-regular, triangle-free graph~\cite{KratzkeW96}. We show however that the problem is polynomial-time decidable if $k$ is a constant.
First, we need a bound on the number of maximal cliques in \kgigs.

\subsection{Maximal Cliques}

A \emph{\clique{}} in a graph $G$ is a set of vertices that are all pairwise adjacent in $G$.
A \clique is \emph{maximal} if it is not a subset of another clique.

\begin{restatable}[\shortversion{$\star$}]{lemma}{Replemubcliques}\label{lem:ubcliques}
 An \ipkvg on $n$ vertices has at most $2^k \cdot (n-k)$ maximal cliques.
\end{restatable}

\longversion{\shortversion{\Replemubcliques*}

\begin{proof}
 Let $G=(V,E)$ be an \ipkvg and $X$ be an \ids of size $k$.
 Let $Y\subseteq X$ and consider all maximal cliques of $G$ whose intersection with $X$ is exactly $Y$.
 For any such maximal clique $S$, $S\setminus Y$ is a maximal clique of $G[\bigcap_{v\in Y} N(v) \setminus X]$.
 As $G[\bigcap_{v\in Y} N(v) \setminus X]$ is an interval graph with at most $n-k$ vertices, there are at most $n-k$ choices for $S \setminus Y$. As there are $2^k$ choices for $Y$, the lemma follows.
\qed \end{proof}

}

On the other hand, Lemma \ref{lem:ubcliques} cannot be substantially improved, even for \kgigs, as there are \kgigs with $\Omega(2^k)$ maximal cliques.
Figure \ref{fig:cliques} represents a \mir with $k$ gaps
of a graph $G=(V,E)$ with vertex set $V=\set{a_0, \cdots, a_k, b_0, \cdots, b_k}$ and an edge between every pair of distinct vertices except $a_i$ and $b_i$, $0\le i\le k$.
Any vertex set containing exactly one of $a_i, b_i$, $0\le i\le k$ forms a maximal clique. Thus, this graph has $2^{k+1}$ maximal cliques.

\begin{figure}[tb]
 \begin{center}
 \begin{tikzpicture}[xscale=1.5,yscale=0.7]
  \tikzset{ivl/.style={thick},
  ivl/.default=black,
  var/.style={midway,above=0.5pt}}
 
  \draw (0,0) node {level $k$};
  \draw[ivl,red,|-] (1,0) -- node[var] {$b_k$} (3.75,0);
  \draw[ivl,|-|] (3.85,0) -- node[var] {$a_k$} (6.75,0);
  \draw[ivl,red,-|] (6.85,0) -- node[var] {$b_k$} (7,0);

  \draw (0,1) node[rotate=90] {$\cdots$};\draw (4,1) node[rotate=90] {$\cdots$};

  \draw (0,2) node {level $2$};
  \draw[ivl,red,|-] (1,2) -- node[var] {$b_2$} (1.35,2);
  \draw[ivl,|-|] (1.45,2) -- node[var] {$a_2$} (4.35,2);
  \draw[ivl,red,-|] (4.45,2) -- node[var] {$b_2$} (7,2);

  \draw (0,3) node {level $1$};
  \draw[ivl,red,|-] (1,3) -- node[var] {$b_1$} (1.15,3);
  \draw[ivl,|-|] (1.25,3) -- node[var] {$a_1$} (4.15,3);
  \draw[ivl,red,-|] (4.25,3) -- node[var] {$b_1$} (7,3);

  \draw (0,4) node {level $0$};
  \draw[ivl,|-|] (1,4) -- node[var] {$a_0$} (3.95,4);
  \draw[ivl,|-|] (4.05,4) -- node[var] {$b_0$} (7,4);
 \end{tikzpicture}
 \end{center}
 \caption{\label{fig:cliques} A \mir with $k$ gaps of a graph with $2^{k+1}$ maximal cliques.}
\end{figure}

\subsection{PQ-trees}

To recognize \kgigs, we make use of PQ-trees. A \emph{PQ-tree} is a rooted tree $T$ that represents allowed permutations over a set\longversion{ of elements} $U$. Each leaf\longversion{ of $T$} corresponds to a unique element of $U$. Internal nodes\longversion{ of $T$} are labeled P or Q. The children of an internal node $v$ appear in a particular order, which can be modified depending on the label of $v$. The order can be reversed if the label is Q, and it can be arbitrarily changed if the label is P. In this way, the tree imposes various restrictions on the order in which the leaves appear.
PQ-trees were famously used to provide a linear-time recognition algorithm for interval graphs~\cite{BoothL76}.

Booth and Lueker~\cite{BoothL76} introduced PQ-trees, together with a \emph{reduction algorithm}. This algorithm, given a PQ-tree $T$ and a collection $\mathbb{S}$ of sets, restricts the set of permutations represented by $T$ to those in which the elements of each $S \in \mathbb{S}$ appear consecutively. It runs in time $O(|U| + |\mathbb{S}| + \sum_{S \in \mathbb{S}}|S|)$.
\longversion{We will use this algorithm in subroutines later.}

Our recognition algorithm for \kgigs will construct a PQ-tree $T$ and add additional constraints to $T$. We describe these constraints now and propose an algorithm to check whether they can be met by $T$.
First, we give some notation.
We say that $u \in U$ is to the \emph{left} of $v \in U$ in $T$ if the order of the leaves induced by $T$ is such that $u$ comes before $v$. We can then also define \emph{right}, \emph{leftmost}, and \emph{rightmost} in a natural way. We say that a set of leaves is \emph{consecutive} in $T$ if they appear consecutively in the order of the leaves induced by the tree.

We now give the type of constraints that we will impose on PQ-trees. 
A PQ-tree $T$ over $U$ \emph{satisfies a partition constraint} 
$B=(i, u^{1}_{L},u^{1}_{R},\ldots,u^{i}_{L}, u^{i}_{R}, S)$, where $\{u^{1}_{L},u^{1}_{R},\ldots,u^{i}_{L}, u^{i}_{R}\} \subseteq S \subseteq U$, if $S$ can be partitioned into $S_{1},\ldots,S_{i}$ such that each $S_{j}$ is consecutive, $u^{j}_{L}$ is the leftmost leaf of $S_{j}$, and $u^{j}_{R}$ is the rightmost leaf of $S_{j}$. 
Moreover, $S_{j}$ appears to the left of $S_{j+1}$ for all $1 \leq j < i$. We use $S_B$ to denote the set $S \backslash \{u^{1}_{L},u^{1}_{R},\ldots,u^{i}_{L}, u^{i}_{R}\}$.

We show that, given a PQ-tree and a set of partition constraints, we can decide in polynomial time whether the leaves of the PQ-tree can be reordered to satisfy these constraints. If so, our algorithm finds the order and the partitions $S_{1},\ldots,S_{i}$ for each of the constraints.

\begin{restatable}[\shortversion{$\star$}]{lemma}{Replempqt}\label{lem:pq-tree-mod}
Let $\mc{Z} = \{B_1,\cdots,B_\ell\}$ be a set of partition constraints such that the sets $S_{B_j}$ are pairwise disjoint. It can be decided in $(|\mc{Z}| \cdot n)^{O(1)} $ time if there exists a valid ordering of the leaves of a PQ-tree $T$ satisfying all constraints in $\mc{Z}$.
\end{restatable}

\longversion{\shortversion{\Replempqt*}

\begin{proof}
For a node $w \in V(T)$, let $T_{w}$ denote the subtree of $T$ rooted at $w$ and let $L_{w}$ denote the set of leaves of $T_{w}$.
We call a partition constraint $(i,\ldots,S)$ \emph{relevant} at $w$ if $1 \leq |S \cap L_{w}| < |S|$. We call a constraint \emph{active} at $w$ if it is relevant at $w$ or at a child of $w$, and \emph{inactive} otherwise. By preprocessing the tree in polynomial time, we can easily decide whether a given constraint is active for a vertex $w \in V(T)$ or not.

We fix a correct ordering bottom-up. Let $w$ be an internal vertex of type P. First, observe that we only need to permute relevant children of $w$. Consider first all relevant children of $w$ that have a leaf $u_{L}^{i}$ or $u_{R}^{i}$ for some $i$ and some constraint as a descendant. We permute these relevant children so that they respect the ordering prescribed by the constraints. 
That is, $u_{L}^{i}$ comes immediate before $u_{R}^{i}$ and before $u_{\cdot}^{j}$ for $j > i$ and after $u_{\cdot}^{j}$ for all $j < i$ for each constraint in $\mc{Z}$. As a consequence a leaf $u_R^i$ in $T_w$ with the corresponding $u_L^i$ not in $T_w$ would have to be the first in the ordering, likewise we might find the last leaf.
If it is not possible to find such an ordering, we answer \textsc{No}. 
Now consider all remaining relevant children. Note that all leaves that are a descendant of such a child must belong to some set $S_{B_j}$, or we can immediately answer \textsc{No}. Since the sets $S_{B_j}$ are disjoint, it is easy to permute them and place them properly within the ordering.

If $w$ is an internal vertex of type Q, we apply the same idea, but there are only two possible orderings to check. If neither helps to satisfy the constraints, we simply answer \textsc{No}.

By applying this procedure in a bottom-up fashion, we can correctly decide whether we can satisfy all constraints.
There are $O(n)$ nodes in the PQ-tree. The ordering of the children can be done in $O(|\mc{Z}| \cdot n^{2})$ time.
\qed \end{proof}

}

\subsection{Recognition Algorithm}

We now show how to use Lemma \ref{lem:pq-tree-mod} to recognize \kgigs.
The algorithm tries to construct a \mir for $G$ with at most $k$ gaps.
It guesses an \ids $X$ for $G$ and a \mir of $G[X]$. Then, it constructs a PQ-tree $T$ for $G\setminus X$ and adds
partition constraints to $T$ that need to be fulfilled by an interval representation of $G\setminus X$ to be merged with
the \mir of $G[X]$. Lemma \ref{lem:pq-tree-mod} can then check whether the guesses led to a \mir of $G$ with $k$ gaps.
We 
refer to Appendix \ref{sec:appOmitted} for the full proof.

\begin{restatable}[\shortversion{$\star$}]{theorem}{Repthmrec} \label{thm:recog}
Given a graph $G$, one can decide whether $I(G) \leq n+k$ in polynomial time if $k$ is a constant.
\end{restatable}

\longversion{\shortversion{\Repthmrec*}

\begin{proof}
As a first step, the algorithm guesses how many intervals are associated to each vertex of $G=(V,E)$, with a total of at most $n+k$ intervals.
There are at most $O(n^{k})$ such choices.
Let $X\subseteq V$ denote the set of vertices to which more than one interval is associated. The algorithm verifies that $G \setminus X$ is an interval graph, otherwise it immediately moves to the next guess.

Then the algorithm enumerates all permutations of the set of all endpoints of intervals associated with vertices in $X$.
There are at most $(4k)!$ such permutations, and
each permutation corresponds to a \mir $f$.
Next, verify that $f$ is indeed a \mir for $G[X]$, otherwise move on to the next permutation.

Consider the set $\mc{K}'$ of maximal cliques of $G \setminus X$. Since $G \setminus X$ is an interval graph, $|\mc{K}'| \leq  |V \setminus X|$ and $\mc{K}'$ can be found in polynomial time using a perfect elimination order~\cite{RoseTL76}. Construct a new set $\mc{K}$ of cliques, where $\mc{K}$ is obtained from $\mc{K}'$ by adding all pairwise intersections of cliques in $\mc{K}'$.
Compute the set $\mc{G}$ of all maximal cliques of $G$. By Lemma~\ref{lem:ubcliques}, constructing $\mc{G}$ takes $2^{k}n^{O(1)}$ time using the polynomial-delay enumeration algorithm of~\cite{TsukiyamaIAS77}.

Consider the \mir $f$ that we have guessed before. 
A \emph{display} of a clique $C$ is a maximal open interval $(a,b)$ such that $(a,b)$ is contained in an interval of each vertex of $C$ and is disjoint from the intervals of all other vertices. 
We say that a clique is \emph{displayable} if it has a display. 
Let $C$ be a displayable clique of $G[X]$.
First we find $\mc{K}_{C} = \{ K \cap (\bigcap_{v \in C} N(v)) \mid K \in \mc{K}'\}$.
For each display of each displayable clique $C$, we guess its \emph{outer cliques}, the two cliques of $\mc{K}_{C}$ that appear first and last under the display. Denote the multiset of chosen cliques by $\mc{K}_{C}^{*} = \{ K^{1}_{C},\bar{K}^{1}_{C},\ldots,K^{t}_{C},\bar{K}^{t}_{C} \}$, numbered in the order in which they were chosen, where $K^{i}_{C}, \bar{K}^{i}_{C}$ are the outer cliques of the $i$-th display. Since $\mc{I}$ has at most $4k$ displays, it takes $n^{O(k)}$ time in total to guess $\mc{K}_{C}^{*}$ for all displayable cliques $C$.

Let $\mc{K}^{*}$ denote the multiset that is the union of all the guessed $\mc{K}_{C}^{*}$, and let $\bar{\mc{K}}$ denote the multiset obtained after adding all cliques that are in $\{ K \setminus X \mid K \in \mc{G}\}$ and not already in $\mc{K}^{*}$. 
Initialize a PQ-tree $T$ with $\bar{\mc{K}}$ as ground set. Run the reduction algorithm of Booth and Lueker~\cite{BoothL76} on $T$ with the set $\mathbb{S} = \{ \{ K \in \bar{\mc{K}} \mid v \in K \} \mid v \in V(G) \setminus X \}$.
If it fails, continue to the next guess of outer cliques.

We build a set $S_{C}$ which we will use for a partition constraint on the PQ-tree conform Lemma~\ref{lem:pq-tree-mod}. 
Initially, $S_{C}$ contains the outer cliques of $C$. 
Consider a maximal clique $K \in \mc{G}$.
Suppose that $K \setminus X$ has a vertex $v$ such that $v$ does not appear in any outer clique of $C = K \cap X$. 
Since $v\in K$, $v$ is a neighbor of every vertex in $C$. Hence the interval of $v$ must be a strict subinterval of a display of $C$. 
Therefore we add $K \setminus X$ to $S_{C}$.
If no such vertex $v$ exists all vertices of $K \setminus X$ appear in the union of the outer cliques of $C$. 
As $K$ is maximal, $K \setminus X$ is not a strict subset of any single outer clique.
All vertices from $K \setminus X$ appear in the union of two outer cliques that are consecutive in the order of their endpoints in $f$, for otherwise an outer clique separates two vertices from $K \setminus X$ and hence contradicts that $K \setminus X$ is a clique.
%
Suppose these consecutive outer cliques are $\bar{K}^{i}_{C},K^{i+1}_{C}$ for some $i$. 
We may assume that $K \setminus X$ contains a vertex from $\bar{K}^{i}_{C} \backslash K^{i+1}_{C}$ and a vertex from $K^{i+1}_{C} \backslash \bar{K}^{i}_{C}$, because otherwise $K \setminus X \subseteq \bar{K}^{i}_{C}$ or $K \setminus X \subseteq K^{i+1}_{C}$. 
But then $K \setminus X$ appears between $\bar{K}^{i}_{C}$ and $K^{i+1}_{C}$, and is not added it to $S_{C}$. 
Suppose instead that the two outer cliques are $K^{i}_{C},\bar{K}^{i}_{C}$. Using similar reasoning, we can then show that $K \setminus X$ appears between $K^{i}_{C}$ and $\bar{K}^{i}_{C}$, and thus must be added to $S_{C}$. 
This finishes the construction of $S_{C}$.

For each displayable clique $C$, we have a set of cliques $S_{C}$ that should appear under a display of $C$. Use Lemma~\ref{lem:pq-tree-mod} to partition them into sets $S^{i}_{C}$, where $S^{i}_{C}$ appears between $K^{i}_{C},\bar{K}^{i}_{C}$. Run the reduction algorithm of Booth and Lueker on $T$ with the set $\mathbb{S}$, which is the union of all $\mathbb{S}_C = \{ S^{i}_{C}, S_{C}^{i} \cup \{K^{i}_{C}\}, S_{C}^{i} \cup \{\bar{K}^{i}_{C}\} \mid 1 \leq i \leq t \}$ taken over all displayable cliques $C$.

Assuming all the above operations succeed, the PQ-tree represents admissible permutations of the maximal cliques that yield a
\mir of $G$ with $k$ gaps.
We obtain this \mir by going from left to right through the leafs of the PQ-tree. 
When going from one clique $K$ to the next $K'$, we close all intervals for vertices in $K \setminus K'$ and we open all intervals for vertices in $K'\setminus K$.

This concludes the proof.
\qed \end{proof}
}

\shortversion{
\begin{proof}[Sketch]
As a first step, the algorithm guesses an \ids $X$ of $G$ of size at most $k$
and it guesses the number of intervals that are assigned to each vertex of $X$, such that the
total number of intervals is at most $|X|+k$. In total there are $O(n^{k})$ choices.
For each choice, the algorithm checks that $G\setminus X$ is an interval graph,
because otherwise it can immediately move to the next choice for $X$.
The algorithm also guesses
the order of all endpoints of intervals associated with vertices in $X$.
There are at most $(4k)!$ different permutations.
The ordering defines a \mir $f$ of $G[X]$ and determines the way the vertices of $X$ overlap.
If this ordering does not match with the edges of $G[X]$, disregard the current guess.

As $G \setminus X$ is a interval graph we can find all the maximal cliques in polynomial time using a perfect elimination order~\cite{RoseTL76}. 
We also find all maximal cliques of $G$ using Lemma~\ref{lem:ubcliques} and a polynomial delay enumeration algorithm \cite{TsukiyamaIAS77}.

Suppose $f$ can be extended into a \mir $f'$ for $G$ by assigning exactly one interval to each vertex from $V\setminus X$.
Consider some endpoint $p$ of an interval in $f$. Then, $p$ defines a clique of
$G \setminus X$, contained within the neighborhoods of some vertices from $X$.
For each endpoint of an interval in $f$, the algorithm guesses this clique and the clique that comes just before $p$.
Build a PQ-tree of the maximal cliques of $G$ restricted to $G \setminus X$ plus the cliques corresponding to endpoints of intervals in $f$.
Then, partition all the cliques in the PQ-tree into sets depending on what subset of intervals from $f$ they will belong to.

Finally we use this partition to add partition constraints to the PQ-tree and apply Lemma \ref{lem:pq-tree-mod}.
Once we have the order of the cliques in the PQ-tree a \mir with $k$ gaps can easily be obtained.
\qed \end{proof}
}

\section{\FPT\ Results}

\noindent
The \MIT problem is specific to multiple interval graphs.
This problem and its variants is well studied for $t$-interval graphs
(see for example \cite{Alon98,HassinS08,Kaiser97,Tardos95}).
Given a graph $G$, a \mir $f$ of $G$, and a positive integer $p$,
the problem asks whether there is a set $P$ of $p$ points on the real line such
that
each vertex of $G$ is associated to an interval containing a point from $P$.
By relating this problem to a problem from Constraint Satisfaction, we obtain
the following result.

\begin{theorem}
 The \MIT problem, parameterized by $k$ has a $O(k^2)$-vertex kernel and can be
solved in time $O(1.6181^k k^2 +n)$ on \kgigs, where $n$ is the number of
vertices of the input graph.
\end{theorem}
\begin{proof}
 \newcommand{\dom}{\mathit{dom}}
 The \textsc{Consistency} problem for AtMost-NValue contraints has as input a
set of variables $X=\set{x_1, \ldots, x_{n'}}$,
 a totally ordered set of values $D$, a map $\dom: X \rightarrow 2^D$ assigning
a non-empty domain $\dom(x)\subseteq D$ to each variable
 $x\in X$, and an integer $N$. The question is whether there exists an
assignment of the variables from $X$ to values from their domain
 such that the number of distinct values taken by variables from $X$ is at most
$N$.

 Bessi{\`e}re \etal \cite{BessiereHHKQW08} were the first to parameterize this
problem by the total number $k'$ of holes in the domains of the variables. Here,
 a hole in the domain of a variable $x$ is a couple $(u,w) \in \dom(x) \times
\dom(x)$, such that there is a value $v \in D \setminus \dom(x)$ with $u<v<w$
and there
 is no value $v'\in \dom(x)$ with $u<v'<w$. The problem has a kernel with
$O(k'^2)$ variables and domain values and can be solved in time $O(1.6181^{k'}
k'^2 +n'+|D|)$ \cite{GaspersS11}.

 The theorem will follow by a simple reduction of a \MIT instance
$(G=(V,E),f,p)$ with parameter $k$ to an instance $(X,D,\dom,N)$ with parameter
$k'=k$ of the \textsc{Consistency} problem for AtMost-NValue contraints.
 Let $F:= \set{l,r : [l,r]\in f(v), v\in V}$ denote the set of all left and
right endpoints of intervals in $f$.
 The reduction sets $X:=V$, $D:=F$, $\dom(x):= \bigcup_{I\in f(x)} I \cap F$,
and $N:=p$. It is easy to see that both instances are equivalent and that
$k'=k$.
\qed \end{proof}

\noindent
A vertex subset $U$ is a \emph{\fvs{}} in a graph $G$ if $G\setminus U$ has no
cycle.
The \FVS problem has as input a graph $G$ and a positive integer $p$, and the
question is whether $G$ has a \fvs of size at most $p$. 

\begin{theorem}
 \FVS can be solved in time  $2^{O(k\log k)}\cdot n^{O(1)}$
 on \ipkvgs with $n$ vertices.
\end{theorem}
\begin{proof}
We design a dynamic-programming algorithm to solve \FVS on \ipkvgs.
The key observation is that any \fvs misses at most two vertices of any clique
of $G$.

Any interval graph (see e.g.~\cite{Golumbic80}) has a path decomposition whose
set of bags is exactly the set of maximal cliques.
Kloks~\cite{Kloks94} showed that every path decomposition of a graph $G$ can be
converted in linear time to a \emph{nice path decomposition}, such that the size
of the largest bag does not increase, and the total size of the path is linear
in the size of the original path. A path decomposition $(B,P)$ is \emph{nice} if
$P$ is a path with nodes $1,\ldots,r$
such that the nodes of $P$ are of two types:
\begin{enumerate}
 \item an \emph{introduce node} $i$ with $B_i=B_{i-1}\cup\{v\}$ for some vertex
$v\in V$ (we assume that $X_0=\emptyset$) ;
 \item a \emph{forget node} $i$ with $B_i=B_{i-1}\setminus\{v\}$ for some vertex
$v\in V$.
\end{enumerate}  
Thus, an interval graph $G$ has a nice path decomposition with the
additional property that each bag is a 
clique in $G$.

Now we are ready to describe our algorithm for \FVS. 
Let $G$ be an \ipkvg with \ids $X$.
Using an interval representation of $G' = G\setminus X$, we construct a path decomposition of $G'$ whose set of bags is the
set of maximal cliques of $G'$, and then we construct in linear time a nice path
decomposition $(B',P')$ of $G'$ where $P'$ is a path on nodes $1,\ldots,r$. Set
$B'_0:=\emptyset$. We construct a path decomposition of $G$ with bags
$B_0,\ldots,B_r$ where $B_i=B'_i\cup X$ for $i\in\{0,\ldots,r\}$.
Now we apply a dynamic programming algorithm over this path decomposition.

We first describe what is stored in the tables corresponding to the nodes $0,\ldots,r$ of
the path. 
For any $i\in\{0,\ldots,r\}$, we denote by $G_i$ the subgraph of $G$ induced by 
$\cup_{j=0}^i B_j$. For $i\in\{0,\ldots,r\}$, the table stores the records
$R=(F,F_i,{\cal P},s)$, where 
\begin{itemize}
\item $F\subseteq X$;
\item $F_i\subseteq B'_i$;
\item $\cal P$ is a partition of $B_i\setminus (F\cup F_i)$; and
\item $s\leq n$ is a positive integer;
\end{itemize}
with the property that there is a \fvs $U_i$ of $G_i$ such that
\begin{itemize}
\item $|U_i|\leq s$;
\item $U_i\cap X=F$ and $U_i\cap B'_i=F_i$;
\item for any set $S$ in $\cal P$, $x,y\in S$ if and only if $x,y$ are in the
same component of $G_i\setminus U_i$. 
\end{itemize}
Clearly, $G$ has a \fvs of size at most $p$ if and only if the
table for $r$ contains a record $R$ with $s=p$.
The tables are created and maintained in a straightforward way. 

It remains to estimate the running time. Since $|X|\leq k$, there are at most
$2^k$ subsets $F$ of $X$. Each $B'_i$ is a clique. Hence, $|B'_i\setminus F_i|\leq
2$,
since otherwise $G_i\setminus U_i$ has a cycle. It follows that we consider at
most $\frac{1}{2}n(n+1)+1$ sets $F_i$. Each set $B_i\setminus (F\cup F_i)$ has
size at most $k+2$, and the number of partitions is upper bounded by $B_{k+2}$,
where $B_t$ is the $t^{\text{th}}$ Bell number.
Finally, $s$ can have at most $n$ values. We conclude that for each  
 $i\in\{0,\ldots,r\}$, the table for $i$ contains at most $O(2^kB_{k+2}\cdot
n^3)$ records. It follows that our algorithm runs in time
$2^{O(k\log k)}\cdot n^{O(1)}$.
\qed \end{proof}

\noindent
A \emph{\cc{}} of size $t$ of a graph $G=(V,E)$ is a partition of $V$
into $Z_1,Z_2,\ldots,Z_t$ where $Z_i$ is a clique in $G$, for $1 \leq i \leq t$. 
The \CC problem has as input a graph $G$ and a positive integer $p$,
and the question is whether $G$ has a \cc of size $p$.

\begin{theorem}
 \CC can be solved in time  $O(2^{k} \cdot n^{O(1)})$ and polynomial space
 on \ipkvgs with $n$ vertices.
\end{theorem}
\begin{proof}
Before starting we observe that there is a minimum \cc where $Z_1$ is a maximal 
clique of $G$ and in general $Z_i$ is a maximal clique of $G[Z_i \cup Z_{i+1} \cup \dots Z_t]$. 
I.e. stealing a vertex from a higher numbered clique will not increase the number of cliques in the cover. 

Let $G$ be a \ipkvg with \ids $X$. Using an interval representation of $G' = G\setminus X$,
we construct a path decomposition of $G'$ whose set of bags 
$B_1,\ldots,B_r$ is the set of maximal cliques of $G'$ (see e.g.~\cite{Golumbic80}). 
As each bag of the path decomposition corresponds to the vertex set of a maximal clique 
in $G'$, there is a vertex $v \in B_1 \setminus (B_2 \cup B_3 \cup \ldots B_r)$. 

The algorithm considers all choices for the intersection of $X$ with the clique from the \cc containing $v$.
Each such choice is a clique $X_1$ such that $N(v) \subseteq X_1 \subseteq X$.
Given $X_1$ and $v$, the clique $c(X_1,v)$ of the \cc containing $X_1 \cup \{v\}$ can be chosen greedily
by the maximality argument mentioned above. Indeed, there is a unique maximal clique
containing $X_1 \cup \{v\}$: we set $c(X_1,v) := X_1 \cup \{v\} \cup Y_1$,
where $u \in Y_1$ if and only if $u\in B_1$ and $X_1 \subseteq N(u)$.
Let $mcc(G)$ be the size of a minimum \cc for $G$. 
Then $mcc(G) = 1 + \min \{mcc(G[V \setminus c(X_1,v)]) : X_1 \text{ is a clique and } N(v) \subseteq X_1 \subseteq X\}$.
As the $X_1$ minimizing the above equation is one of the $2^k$ subsets of $X$
we can conclude that clique cover is computed correctly in time $O(2^{k} \cdot n^{O(1)})$ and polynomial space.
\qed \end{proof}

\noindent
The \emph{boolean-width} of graphs is a recently introduced graph parameter \cite{BuiXuanTV2011}.
It will enable us to obtain \FPT\ results for several problems.
As interval graphs have boolean-width at most $\log n$ \cite{BelmonteV11} and adding a vertex to
a graph increases its boolean-width by at most 1, we have the following lemma.

\begin{lemma}
 Any \ipkvg $G$ has boolean width at most $\log n + k$, where $n$ is the number of vertices of $G$.
\end{lemma}

\noindent
As several problems can be solved in time $2^{O(b)} n^{O(1)}$ on graphs with
boolean-width $b$ and $n$ vertices \cite{BuiXuanTV2011},
they are \FPT\ on \ipkvgs.

\begin{corollary}
 \IS, \DS, their
 weighted and counting versions, and \textsc{Independent Dominating Set},
 are \FPT\ on \ipkvgs.
\end{corollary}

\noindent
We also provide simple polynomial-space algorithms for \IS and \CLIQUE on \ipkvgs and for \DS on \kgigs{}\shortversion{ in Appendix \ref{sec:appISDS}}.

\longversion{
An \emph{\is{}} in a graph $G$ is a set of vertices that are all pairwise non-adjacent in $G$.
The \IS (\CLIQUE) problem has as input a graph $G$ and a positive integer $p$, and the question is whether $G$ has an \is (\clique) of size $p$.

\begin{theorem}\label{thm:is}
 \IS and \CLIQUE can be solved in time $2^k \cdot n^{O(1)}$ and polynomial space on \ipkvgs, where $n$ is the number of vertices.
\end{theorem}
\begin{proof}
 Let $G=(V,E)$ be the input graph with \ids $X$.

 To check whether $G$ has an \is of size $p$, the algorithm goes over all subsets $Y \subseteq X$,
 and checks whether $Y$ is independent in $G$. If so, it checks whether $G \setminus (X \cup N[Y])$
 has an \is of size $p-|Y|$. The last check can be done in linear time \cite{Gavril74}, as $G \setminus (X \cup N[Y])$
 is an interval graph.

 To check whether $G$ has a \clique of size $p$, the algorithm uses a polynomial-delay polynomial-space algorithm enumerating
 all maximal cliques of $G$ \cite{TsukiyamaIAS77}, and checks whether at least one such maximal clique has size at least $p$.
 As $G$ has $O(2^k n)$ maximal cliques by Lemma \ref{lem:ubcliques}, the running time follows.
\qed \end{proof}

 \noindent
A vertex subset $D$ is a \emph{\ds{}} in a graph $G=(V,E)$ if every vertex from $V \setminus D$ has a neighbor in $D$.
The \DS problem has as input a graph $G$ and a positive integer $p$, and the question is whether $G$ has a \ds of size $p$.

\begin{theorem}
 \DS can be solved in time $O(3^k \cdot n^{O(1)})$ and polynomial space on \kgigs.
\end{theorem}
\begin{proof}
 Let $G=(V,E)$ be the input graph, let $f$ be a \mir of $G$ with $k$ gaps, and let $X=\gap_f(G)$.

 The algorithm goes over all partitions $(Y,Z)$ of $X$. For each such partition it will consider dominating sets
 containing $Y$.
 Each vertex $z\in Z$ needs to be dominated by a vertex that has is associated to an interval intersecting
 at least one interval of $z$.
 For each vertex $z\in Z$, the algorithm considers all possibilities of choosing exactly one interval
 $[l_z,r_z]$ from $f(z)$ by which it is to be dominated. It remains to check whether $G$ has a \ds $D$ of size $p$
 such that $Y\subseteq D$ and such that each $[l_z,r_z]$ intersects at least one interval of a vertex from
 $D$, for each $z\in Z\setminus D$. This is done by a polynomial-time algorithm which solves
 a version of \DS on interval graphs where some
 vertices do not need to be dominated \cite{RamalingamR88}. Namely, we start from the interval model
 $\{f(v) : v\in V\setminus X\} \cup \{[l_z,r_z] : z\in Z\}$, we mark the intervals associated to vertices from $N(Y)$
 and check whether this interval graph has $p-|Y|$ vertices dominating every vertex in $V\setminus N[Y]$
 by the algorithm from \cite{RamalingamR88}.
 
 As choosing an interval for each vertex from $Y$ can be reduced to deciding, for each gap, whether the interval is to
 the left or to the right of this gap, the total running time is within a polynomial factor of
 $\sum_{Y\subseteq X} \prod_{z\in X\setminus Y} 2^{|f(z)|-1} = O(3^k)$.
\qed \end{proof}
}

\section{\W[1]-Hardness Result}



A \emph{coloring} of a graph $G=(V,E)$ is a mapping $c\colon V\rightarrow\{1,2,\ldots\}$ such that $c(u)\neq c(v)$ whenever
$uv\in E$. A \emph{$p$-coloring} of $G$ is a
coloring $c$ of $G$ with $c(v)\in\{1,\ldots,p\}$ for $v\in V$.  
The $p$-{\sc Coloring} problem asks for a graph $G$ and a positive integer $p$, whether $G$
has a $p$-coloring. 
The problem $p$-{\sc Precoloring Extension} is to decide whether a given
mapping $c\colon U\rightarrow\{1,\ldots,p\}$ defined on a
(possibly empty) subset $U\subseteq V$ of \emph{precolored} vertices can
be extended to a $p$-coloring of $G$. We refer to these problems as {\sc Coloring} and {\sc Precoloring Extention} if $p$ is assumed to be a part of the input.

First, we make the following observation that was independently made in~\cite{JansenK11}.

\begin{proposition}\label{prop:colFPT}
The parameterization of {\sc Coloring} by $p+k$ is \FPT\ on \ipkvgs.
\end{proposition}

\begin{proof}
We use a Win-Win approach.
Let $G$ be an \ipkvg with \ids $X$.
If $G$ has a clique of size $p+1$, then it cannot be colored by $p$ colors.
By Theorem~\ref{thm:is} it can be determined whether such a clique exists in time $2^k \cdot n^{O(1)}$.
Otherwise, the interval graph $G\setminus X$ has pathwidth at most $p$~\cite{Bodlaender98}.
Thus, $\pw(G)\leq p+k$.
It remains to observe that $p$-{\sc Coloring} is \FPT\ on graphs of bounded pathwidth by Courcelle's Theorem~\cite{Courcelle92}.
\qed \end{proof}

\noindent
However, the parameterization by $k$ of this problem is \W[1]-hard, even for \kgigs.

\begin{theorem}\label{thm:colW}
 \COL, parameterized by $k$, is \W[1]-hard on \kgigs.
\end{theorem}
\begin{proof}
We reduce from the {\sc Precoloring Extension} problem.
Marx~\cite{Marx06} proved that {\sc Precoloring Extension} 
is \W[1]-hard on interval graphs, parameterized by the number of precolored vertices.
Let $G=(V,E)$ be an interval graph with a set of precolored vertices $U\subseteq V$ and a precoloring $c\colon U\rightarrow\{1,\ldots,p\}$. 
Let $k=|U|$, and denote by $X_1,\ldots,X_p$ (some sets can be empty) the partition of $U$ into the color classes induced by $c$. We construct the graph $H$ as follows:
\begin{itemize}
\item construct a disjoint union of $G$ and a complete graph $K_p$ with the vertices $v_1,\ldots,v_p$;
\item for each $i\in\{1,\ldots,p\}$, identify all the vertices of $X_i$ and $v_i$.
\end{itemize}
By Observation~\ref{obs:ident}, $H$ is a \kgig. It remains to observe that $H$ has a $p$-coloring if and only if $c$ can be extended to a $p$-coloring of $G$.
\qed \end{proof}

\section{Conclusion}

While multiple interval graphs have a large number of applications, many problems remain intractable on $t$-interval graphs, even for small constant $t$.
On the other hand, the total number of gaps, $k$, in a \mir seems to be a more useful parameterization of problems on multiple interval graphs. Indeed, we have
seen that this parameter captures some of the intractibility of graph problems and the parameterization by $k$ of many problems turns out to
be \FPT.

While this first paper on the parameterization of graph problems by the total number of gaps
classifies some important problems as \FPT\ or \W[1]-hard, it raises more questions than it answers. There is the question of investigating other problems
that are polynomial time solvable on interval graphs but hard on $t$-interval graphs for small constant $t$. One example is \textsc{Hamiltonian Cycle}.
Further considerations worth investigating are kernelization algorithms and improvements on the running time of our (rather simple) algorithms.
The most important open problem for \kgigs is, in our eyes, to pinpoint the parameterized complexity of the recognition
problem. 


\emph{Acknowledgment.} We thank Mathieu Chapelle for interesting discussions about this work.

{
\bibliographystyle{serge-short}
\bibliography{shorttitles,kgapig}

\begin{thebibliography}{10}

\bibitem{Alon98}
N.~Alon.
\newblock Piercing $d$-intervals.
\newblock {\em Discret. Comput. Geom.}, 19(3):333--334, 1998.

\bibitem{Andreae86}
T.~Andreae.
\newblock On an extremal problem concerning the interval number of a graph.
\newblock {\em Discrete Appl. Math.}, 14(1):1--9, 1986.

\bibitem{AndreaeA89}
T.~Andreae and M.~Aigner.
\newblock The total interval number of a graph.
\newblock {\em J. Comb. Theory Ser. B}, 46(1):7--21, 1989.

\bibitem{AumannLMPY05}
Y.~Aumann, M.~Lewenstein, O.~Melamud, R.~Y. Pinter, and Z.~Yakhini.
\newblock Dotted interval graphs and high throughput genotyping.
\newblock In {\em SODA 2005},  339--348, 2005.

\bibitem{BafnaNR96}
V.~Bafna, B.~O. Narayanan, and R.~Ravi.
\newblock Nonoverlapping local alignments (weighted independent sets of
  axis-parallel rectangles).
\newblock {\em Discrete Appl. Math.}, 71(1-3):41--53, 1996.

\bibitem{BaloghOP04}
J.~Balogh, P.~Ochem, and A.~Pluh{\'a}r.
\newblock On the interval number of special graphs.
\newblock {\em J. Graph Theor.}, 46(4):241--253, 2004.

\bibitem{Bar-YehudaHNSS06}
R.~Bar-Yehuda, M.~M. Halld{\'o}rsson, J.~Naor, H.~Shachnai, and I.~Shapira.
\newblock Scheduling split intervals.
\newblock {\em SIAM J. Comput.}, 36(1):1--15, 2006.

\bibitem{Bar-YehudaR06}
R.~Bar-Yehuda and D.~Rawitz.
\newblock Using fractional primal-dual to schedule split intervals with
  demands.
\newblock {\em Discrete Optim.}, 3(4):275--287, 2006.

\bibitem{BelmonteV11}
R.~Belmonte and M.~Vatshelle.
\newblock Graph classes with structured neighborhoods and algorithmic
  applications.
\newblock In {\em WG 2011}, LNCS 6986,  47--58, 2011.

\bibitem{BessiereHHKQW08}
C.~Bessi{\`e}re, E.~Hebrard, B.~Hnich, Z.~Kiziltan, C.-G. Quimper, and
  T.~Walsh.
\newblock The parameterized complexity of global constraints.
\newblock In {\em AAAI 2008},  235--240, 2008.

\bibitem{BlinFV07}
G.~Blin, G.~Fertin, and S.~Vialette.
\newblock Extracting constrained 2-interval subsets in 2-interval sets.
\newblock {\em Theor. Comput. Sci.}, 385(1-3):241--263, 2007.

\bibitem{Bodlaender98}
H.~L. Bodlaender.
\newblock A partial $k$-arboretum of graphs with bounded treewidth.
\newblock {\em Theor. Comput. Sci.}, 209(1-2):1--45, 1998.

\bibitem{BoothL76}
K.~S. Booth and G.~S. Lueker.
\newblock Testing for the consecutive ones property, interval graphs, and graph
  planarity using pq-tree algorithms.
\newblock {\em J. Comput. System Sci.}, 13(3):335--379, 1976.

\bibitem{BuiXuanTV2011}
B.-M. Bui-Xuan, J.~A. Telle, and M.~Vatshelle.
\newblock Boolean-width of graphs.
\newblock {\em Theor. Comput. Sci.}, 412(39):5187--5204, 2011.

\bibitem{ButmanHLR10}
A.~Butman, D.~Hermelin, M.~Lewenstein, and D.~Rawitz.
\newblock Optimization problems in multiple-interval graphs.
\newblock {\em ACM Trans. Algorithms}, 6(2), 2010.

\bibitem{Catlin92}
P.~A. Catlin.
\newblock Supereulerian graphs: A survey.
\newblock {\em J. Graph Theor.}, 16(2):177--196, 1992.

\bibitem{ChenYY07}
E.~Chen, L.~Yang, and H.~Yuan.
\newblock Improved algorithms for largest cardinality 2-interval pattern
  problem.
\newblock {\em J. Comb. Optim.}, 13(3):263--275, 2007.

\bibitem{ChenC02}
M.~Chen and G.~J. Chang.
\newblock Total interval numbers of complete $r$-partite graphs.
\newblock {\em Discrete Appl. Math.}, 122:83--92, 2002.

\bibitem{Courcelle92}
B.~Courcelle.
\newblock The monadic second-order logic of graphs {III}: tree-decompositions,
  minor and complexity issues.
\newblock {\em Rairo - Theor. Inform. Appl.}, 26:257--286, 1992.

\bibitem{CrochemoreHLRV08}
M.~Crochemore, D.~Hermelin, G.~M. Landau, D.~Rawitz, and S.~Vialette.
\newblock Approximating the 2-interval pattern problem.
\newblock {\em Theor. Comput. Sci.}, 395(2-3):283--297, 2008.

\bibitem{DowneyF99}
R.~G. Downey and M.~R. Fellows.
\newblock {\em Parameterized complexity}.
\newblock Springer, 1999.

\bibitem{ErdosW85}
P.~Erd{\"o}s and D.~B. West.
\newblock A note on the interval number of a graph.
\newblock {\em Discrete Math.}, 55(2):129--133, 1985.

\bibitem{FellowsHRV09}
M.~R. Fellows, D.~Hermelin, F.~Rosamond, and S.~Vialette.
\newblock On the parameterized complexity of multiple-interval graph problems.
\newblock {\em Theor. Comput. Sci.}, 410:53--61, 2009.

\bibitem{FlumG06}
J.~Flum and M.~Grohe.
\newblock {\em Parameterized Complexity Theory}, Texts in Theoretical Computer
  Science. An EATCS Series XIV.
\newblock Springer, 2006.

\bibitem{GambetteV07}
P.~Gambette and S.~Vialette.
\newblock On restrictions of balanced 2-interval graphs.
\newblock In {\em WG 2007}, LNCS 4769,  55--65, 2007.

\bibitem{GaspersS11}
S.~Gaspers and S.~Szeider.
\newblock Kernels for global constraints.
\newblock In {\em IJCAI 2011},  540--545, 2011.

\bibitem{Gavril74}
F.~Gavril.
\newblock The intersection graphs of subtrees in trees are exactly the chordal
  graphs.
\newblock {\em J. Comb. Theory Ser. B}, 16(1):47--56, 1974.

\bibitem{Golumbic80}
M.~C. Golumbic.
\newblock {\em Algorithmic graph theory and perfect graphs}.
\newblock Academic Press, 1980.

\bibitem{GriggsW80}
J.~R. Griggs and D.~B. West.
\newblock Extremal values of the interval number of a graph.
\newblock {\em SIAM J. Algebra. Discr.}, 1(1):1--7, 1980.

\bibitem{HassinS08}
R.~Hassin and D.~Segev.
\newblock Rounding to an integral program.
\newblock {\em Oper. Res. Lett.}, 36(3):321--326, 2008.

\bibitem{JansenK11}
B.~M.~P. Jansen and S.~Kratsch.
\newblock Data reduction for graph coloring problems.
\newblock In {\em FCT 2011}, LNCS 6914,  90--101, 2011.

\bibitem{JiangZ11b}
M.~Jiang and Y.~Zhang.
\newblock Parameterized complexity in multiple-interval graphs: domination.
\newblock In {\em IPEC 2011}, LNCS 7112,  27--40, 2011.

\bibitem{JiangZ11}
M.~Jiang and Y.~Zhang.
\newblock Parameterized complexity in multiple-interval graphs: partition,
  separation, irredundancy.
\newblock In {\em COCOON 2011}, LNCS 6842,  62--73, 2011.

\bibitem{Kaiser97}
T.~Kaiser.
\newblock Transversals of d-intervals.
\newblock {\em Discret. Comput. Geom.}, 18(2), 1997.

\bibitem{Kloks94}
T.~Kloks.
\newblock {\em Treewidth, Computations and Approximations}, LNCS 842.
\newblock Springer, 1994.

\bibitem{KostochkaW97}
A.~V. Kostochka and D.~B. West.
\newblock Total interval number for graphs with bounded degree.
\newblock {\em J. Graph Theor.}, 25(1):79--84, 1997.

\bibitem{KratzkeW93}
T.~M. Kratzke and D.~B. West.
\newblock The total interval number of a graph, {I}: Fundamental classes.
\newblock {\em Discrete Math.}, 118(1-3):145--156, 1993.

\bibitem{KratzkeW96}
T.~M. Kratzke and D.~B. West.
\newblock The total interval number of a graph {II}: Trees and complexity.
\newblock {\em SIAM J. Discrete Math.}, 9(2):339--348, 1996.

\bibitem{Marx06}
D.~Marx.
\newblock Parameterized coloring problems on chordal graphs.
\newblock {\em Theor. Comput. Sci.}, 351(3):407--424, 2006.

\bibitem{Marx10}
D.~Marx.
\newblock Chordal deletion is fixed-parameter tractable.
\newblock {\em Algorithmica}, 57(4):747--768, 2010.

\bibitem{Niedermeier06}
R.~Niedermeier.
\newblock {\em Invitation to Fixed-Parameter Algorithms}.
\newblock Oxford Lecture Series in Mathematics and Its Applications. Oxford
  University Press, 2006.

\bibitem{RamalingamR88}
G.~Ramalingam and C.~Pandu~Rangan.
\newblock A unified approach to domination problems on interval graphs.
\newblock {\em Inform. Process. Lett.}, 27(5):271--274, 1988.

\bibitem{Raychaudhuri95}
A.~Raychaudhuri.
\newblock The total interval number of a tree and the hamiltonian completion
  number of its line graph.
\newblock {\em Inform. Process. Lett.}, 56(6):299--306, 1995.

\bibitem{RoseTL76}
D.~J. Rose, R.~E. Tarjan, and G.~S. Lueker.
\newblock Algorithmic aspects of vertex elimination on graphs.
\newblock {\em SIAM J. Comput.}, 5(2):266--283, 1976.

\bibitem{ScheinermanW83}
E.~R. Scheinerman and D.~B. West.
\newblock The interval number of a planar graph: Three intervals suffice.
\newblock {\em J. Comb. Theory Ser. B}, 35(3):224--239, 1983.

\bibitem{Spinrad03}
J.~P. Spinrad.
\newblock {\em Efficient Graph Representations}, Fields Institute
  Monographs~19.
\newblock AMS, 2003.

\bibitem{Tardos95}
G.~Tardos.
\newblock Transversals of 2-intervals, a topological approach.
\newblock {\em Combinatorica}, 15(1):123--134, 1995.

\bibitem{TrotterH79}
W.~T. Trotter and F.~Harary.
\newblock On double and multiple interval graphs.
\newblock {\em J. Graph Theor.}, 3(3):205--2011, 1979.

\bibitem{TsukiyamaIAS77}
S.~Tsukiyama, M.~Ide, H.~Ariyoshi, and I.~Shirakawa.
\newblock A new algorithm for generating all the maximal independent sets.
\newblock {\em SIAM J. Comput.}, 6(3):505--517, 1977.

\bibitem{Vialette04}
S.~Vialette.
\newblock On the computational complexity of 2-interval pattern matching
  problems.
\newblock {\em Theor. Comput. Sci.}, 312(2-3):224--239, 2004.

\bibitem{Vialette08}
S.~Vialette.
\newblock Two-interval pattern problems.
\newblock In {\em Encyclopedia of Algorithms}. Springer, 2008.

\bibitem{West89}
D.~B. West.
\newblock A short proof of the degree bound for interval number.
\newblock {\em Discrete Math.}, 73(3):309--310, 1989.

\bibitem{WestS84}
D.~B. West and D.~B. Shmoys.
\newblock Recognizing graphs with fixed interval number is {NP}-complete.
\newblock {\em Discrete Appl. Math.}, 8:295--305, 1984.

\end{thebibliography}
}

\shortversion{
 \appendix\newpage
 
 \section{Omitted Proofs}
 \label{sec:appOmitted}
 \shortversion{\Replemubcliques*}

\begin{proof}
 Let $G=(V,E)$ be an \ipkvg and $X$ be an \ids of size $k$.
 Let $Y\subseteq X$ and consider all maximal cliques of $G$ whose intersection with $X$ is exactly $Y$.
 For any such maximal clique $S$, $S\setminus Y$ is a maximal clique of $G[\bigcap_{v\in Y} N(v) \setminus X]$.
 As $G[\bigcap_{v\in Y} N(v) \setminus X]$ is an interval graph with at most $n-k$ vertices, there are at most $n-k$ choices for $S \setminus Y$. As there are $2^k$ choices for $Y$, the lemma follows.
\qed \end{proof}

 \section{Polynomial-Space Algorithms for \IS, \CLIQUE, and \DS}
 \label{sec:appISDS}
An \emph{\is{}} in a graph $G$ is a set of vertices that are all pairwise non-adjacent in $G$.
The \IS (\CLIQUE) problem has as input a graph $G$ and a positive integer $p$, and the question is whether $G$ has an \is (\clique) of size $p$.

\begin{theorem}\label{thm:is}
 \IS and \CLIQUE can be solved in time $2^k \cdot n^{O(1)}$ and polynomial space on \ipkvgs, where $n$ is the number of vertices.
\end{theorem}
\begin{proof}
 Let $G=(V,E)$ be the input graph with \ids $X$.

 To check whether $G$ has an \is of size $p$, the algorithm goes over all subsets $Y \subseteq X$,
 and checks whether $Y$ is independent in $G$. If so, it checks whether $G \setminus (X \cup N[Y])$
 has an \is of size $p-|Y|$. The last check can be done in linear time \cite{Gavril74}, as $G \setminus (X \cup N[Y])$
 is an interval graph.

 To check whether $G$ has a \clique of size $p$, the algorithm uses a polynomial-delay polynomial-space algorithm enumerating
 all maximal cliques of $G$ \cite{TsukiyamaIAS77}, and checks whether at least one such maximal clique has size at least $p$.
 As $G$ has $O(2^k n)$ maximal cliques by Lemma \ref{lem:ubcliques}, the running time follows.
\qed \end{proof}

 \noindent
A vertex subset $D$ is a \emph{\ds{}} in a graph $G=(V,E)$ if every vertex from $V \setminus D$ has a neighbor in $D$.
The \DS problem has as input a graph $G$ and a positive integer $p$, and the question is whether $G$ has a \ds of size $p$.

\begin{theorem}
 \DS can be solved in time $O(3^k \cdot n^{O(1)})$ and polynomial space on \kgigs.
\end{theorem}
\begin{proof}
 Let $G=(V,E)$ be the input graph, let $f$ be a \mir of $G$ with $k$ gaps, and let $X=\gap_f(G)$.

 The algorithm goes over all partitions $(Y,Z)$ of $X$. For each such partition it will consider dominating sets
 containing $Y$.
 Each vertex $z\in Z$ needs to be dominated by a vertex that has is associated to an interval intersecting
 at least one interval of $z$.
 For each vertex $z\in Z$, the algorithm considers all possibilities of choosing exactly one interval
 $[l_z,r_z]$ from $f(z)$ by which it is to be dominated. It remains to check whether $G$ has a \ds $D$ of size $p$
 such that $Y\subseteq D$ and such that each $[l_z,r_z]$ intersects at least one interval of a vertex from
 $D$, for each $z\in Z\setminus D$. This is done by a polynomial-time algorithm which solves
 a version of \DS on interval graphs where some
 vertices do not need to be dominated \cite{RamalingamR88}. Namely, we start from the interval model
 $\{f(v) : v\in V\setminus X\} \cup \{[l_z,r_z] : z\in Z\}$, we mark the intervals associated to vertices from $N(Y)$
 and check whether this interval graph has $p-|Y|$ vertices dominating every vertex in $V\setminus N[Y]$
 by the algorithm from \cite{RamalingamR88}.
 
 As choosing an interval for each vertex from $Y$ can be reduced to deciding, for each gap, whether the interval is to
 the left or to the right of this gap, the total running time is within a polynomial factor of
 $\sum_{Y\subseteq X} \prod_{z\in X\setminus Y} 2^{|f(z)|-1} = O(3^k)$.
\qed \end{proof}
}

\end{document}